\documentclass[11pt,a4paper]{article}
\usepackage{fullpage}
\usepackage{amssymb,latexsym,amsmath,amsfonts,amsthm}
\usepackage{graphicx}
\usepackage{color}

\usepackage{overpic}
\usepackage{amsmath}
\usepackage{amsthm}
\usepackage{amssymb}
\usepackage{amsfonts}
\usepackage{hyperref}
\usepackage{indentfirst}
\usepackage{enumitem}

 \DeclareMathOperator{\Res}{Res}

\DeclareMathOperator{\pv}{p.v.}
\DeclareMathOperator{\Ai}{Ai}
\DeclareMathOperator{\local}{loc}

\DeclareMathOperator{\curved}{cur}
\DeclareMathOperator{\vertical}{ver}

\newcommand{\realR}{\mathbb{R}}

\renewcommand{\Re}{\mathrm{Re}\,}
\renewcommand{\Im}{\mathrm{Im}\,}

\renewcommand{\vec}{\mathbf}

\newcommand{\loc}{\mathrm{loc}}

\newcommand{\ud}{\,\mathrm{d}}
\newtheorem{thm}{Theorem}[section]
\newtheorem{cor}[thm]{Corollary}

\newtheorem{prop}[thm]{Proposition}

\theoremstyle{remark}

\numberwithin{equation}{section}

\newtheorem{remark}{\indent \textrm{Remark}}[section]

\newcommand{\bigO}{\mathcal{O}}

\newcommand{\eq}{\begin{equation}}
\newcommand{\nq}{\end{equation}}
\newcommand{\eqa}{\begin{eqnarray}}
\newcommand{\nqa}{\end{eqnarray}}

\begin{document}

\title{Local universality in biorthogonal Laguerre ensembles}
\author{Lun Zhang\footnotemark[1]}
\maketitle
\renewcommand{\thefootnote}{\fnsymbol{footnote}}
\footnotetext[1] {School of Mathematical Sciences and Shanghai Key Laboratory for Contemporary Applied Mathematics, Fudan University, Shanghai 200433, P. R. China. E-mail:
lunzhang\symbol{'100}fudan.edu.cn}

\begin{abstract}
We consider $n$ particles $0\leq x_1<x_2< \cdots < x_n < +\infty$, distributed according to a probability measure of the form
$$
\frac{1}{Z_n}\prod_{1\leq i <j \leq n}(x_j-x_i)\prod_{1\leq i <j \leq n}(x_j^{\theta}-x_i^{\theta})\prod_{j=1}^nx_j^\alpha e^{-x_j}\ud x_j, ~~ \alpha>-1,~~ \theta>0,
$$
where $Z_n$ is the normalization constant. This distribution arises in the context of modeling disordered conductors in the metallic regime, and can also be realized as the distribution for squared singular values of certain triangular random matrices. We give a double contour integral formula for the correlation kernel, which allows us to establish universality for the local statistics of the particles, namely, the bulk universality and the soft edge universality via the sine kernel and the Airy kernel, respectively. In particular, our analysis also leads to new double contour integral representations of scaling limits at the origin (hard edge), which are equivalent to those found in the classical work of Borodin. We conclude this paper by relating the correlation kernels to those appearing in recent studies of products of $M$ Ginibre matrices for the special cases $\theta=M\in\mathbb{N}$.
\end{abstract}


\section{Introduction and statement of the main results} \label{sec:Intro}

\subsection{Biorthogonal Laguerre ensembles}
The biorthogonal Laguerre ensembles refer to  $n$ particles $x_1< \cdots < x_n$ distributed over the positive real axis, following a probability density function of the form
\begin{equation}\label{eq:bioLag}
\frac{1}{Z_n}\Delta(x_1,\ldots,x_n)\Delta(x_1^\theta,\ldots,x_n^\theta)\prod_{j=1}^nx_j^\alpha e^{-x_j}, \qquad \alpha>-1, \qquad \theta >0,
\end{equation}
where
$$
Z_n=Z_n(\alpha,\theta)=\int_{[0,\infty)^n}\Delta(x_1,\ldots,x_n)\Delta(x_1^\theta,\ldots,x_n^\theta)\prod_{j=1}^nx_j^\alpha e^{-x_j}\ud x_j
$$
is the normalization constant, and
$$
\Delta(\lambda_1,\ldots,\lambda_n)=\prod_{1\leq i <j \leq n}(\lambda_j-\lambda_i)
$$
is the standard Vandermonde determinant.

Densities of the form \eqref{eq:bioLag} were first introduced by Muttalib \cite{Mut95}, where he pointed out that, due to the appearance of two body interaction term $\Delta(x_1,\ldots,x_n)\Delta(x_1^\theta,\ldots,x_n^\theta)$, these ensembles provide more effective description of disordered conductors in the metallic regime than the classical random matrix theory. A more concrete physical example that leads to \eqref{eq:bioLag} (with $\theta=2$) can be found in \cite{LSZ}, where the authors proposed a random matrix model for disordered bosons. These ensembles are further studied by Borodin \cite{Borodin99} under a more general framework, namely, biorthogonal ensembles. It is also worthwhile to mention the work of Cheliotis \cite{Cheliotis14}, where the author constructed certain triangular random matrices in terms of a Wishart matrix whose squared singular values are distributed according to \eqref{eq:bioLag}; see also \cite{Forrester-Wang15}. Note that when $\theta=1$, \eqref{eq:bioLag} reduces to the well-known Wishart-Laguerre unitary ensemble and plays a fundamental role in random matrix theory; cf. \cite{Anderson-Guionnet-Zeitouni10,Forrester10}.

A nice property of \eqref{eq:bioLag} is that, as proved in \cite{Mut95}, they form the so-called determinantal point processes \cite{Johansson06,Soshnikov00}. This means there exits a correlation kernel $K_n^{(\alpha, \theta)}(x,y)$ such that the joint probability density functions \eqref{eq:bioLag} can be rewritten as the following determinantal forms
\begin{equation*}
\frac{1}{n!}\det\left(K_n^{(\alpha,\theta)}(x_i,x_j)\right)_{i,j=1}^n.
\end{equation*}
The kernel $K_n^{(\alpha, \theta)}(x,y)$ has a representation in terms of the so-called biorthogonal polynomials (cf. \cite{Kon65} for a definition). Let
\begin{equation}
p_j^{(\alpha,\theta)}(x)=\kappa_j x^j+\ldots, \qquad q_k^{(\alpha,\theta)}(x)= x^k+\ldots, \quad \kappa_j > 0,
\end{equation}
be two sequences of polynomials depending on the parameters $\alpha$ and $\theta$, of degree $j$ and $k$ respectively, and they satisfy the orthogonality conditions
\begin{equation}\label{eq:bioOP}
\int_0^\infty p_j^{(\alpha,\theta)}(x)q_k^{(\alpha,\theta)}(x^\theta)x^\alpha e^{-x}\ud x=\delta_{j,k}, \qquad j,k=0,1,2,\ldots.
\end{equation}
Note that the polynomial $q_k^{(\alpha, \theta)}$ is normalized to be monic. We then have
\begin{equation}\label{eq:kerBio}
K_n^{(\alpha, \theta)}(x,y)=\sum_{j=0}^{n-1}p_j^{(\alpha, \theta)}(x)q_j^{(\alpha, \theta)}(y^\theta )
x^\alpha e^{-x}.
\end{equation}

The families $\{p_j,j=0,1,\ldots\}$ and $\{q_k,k=0,1,\ldots\}$, which are called Laguerre biorthogonal polynomials,
exist uniquely, since the associated bimoment matrix is nonsingular; see \eqref{def:bimom} and \eqref{eq:detBiMon} below.
The studies of these polynomials (with $\theta=2$) might be traced back to \cite{SF51} during the investigations of
penetration and diffusion of x-rays through matter. Later, intensive studies have been conducted on the case $\theta\in\mathbb{N}=\{1,2,\ldots\}$ in \cite{Carl,GC69a,GC69b,Kon67,Pra70,Pre62,Sri73}, where the general properties including explicit formulas, recurrence relations, generating functions, Rodrigues's formulas etc. are derived.

As determinantal point processes, a fundamental issue of the study is to establish the large $n$ limit of the correlation kernel \eqref{eq:kerBio} in both macroscopic and microscopic regimes. By expressing $K_n^{(\alpha,\theta)}(x,y)$ as a finite series expansion in terms of $x^{k\theta}y^r$, $k,r=0,1,\ldots,n-1$, it was shown by Borodin \cite[Theorem 4.2]{Borodin99} that
\begin{align}\label{eq:BorHard}
\lim_{n \to \infty}\frac{K_n^{(\alpha, \theta)}\left(\frac{x}{n^{1/\theta}},\frac{y}{n^{1/\theta }}\right)}{n^{1/\theta}}
&=\sum_{k,l=0}^{\infty}\frac{(-1)^kx^{\alpha+k}}{k!\Gamma\left(\frac{\alpha+1+k}{\theta}\right)}\frac{(-1)^ly^{\theta l}}{l!\Gamma(\alpha+1+\theta l)}\frac{\theta}{\alpha+1+k+\theta l}
\nonumber \\
&=\theta x^{\alpha}\int_0^1J_{\frac{\alpha+1}{\theta},\frac{1}{\theta}}(ux)J_{\alpha+1,\theta}((uy)^{\theta})u^\alpha\ud u,
\end{align}
where $J_{a,b}$ is Wright's generalization of the Bessel function \cite{Er53} given by
\begin{equation}\label{eq:Wright}
J_{a,b}(x)=\sum_{j=0}^\infty\frac{(-x)^j}{j!\Gamma(a+bj)}; \end{equation}
see also \cite{LSZ} for the special case $\theta=2$, $\alpha \in \mathbb{N}\cup \{0\}$.
These non-symmetric hard edge scaling limits generalize the classical Bessel kernels \cite{Forrester93,TW94} (corresponding to $\theta=1$), and possess some nice symmetry properties. Moreover, they also appear in the studies of large $n$ limits of correlation kernels for biorthogonal Jacobi and biorthogonal Hermite ensembles \cite{Borodin99}. When $\theta=M\in\mathbb{N}$ or $1/\theta = M$, the limiting kernels coincide with the hard edge scaling limits of specified parameters arising from products of $M$ Ginibre matrices \cite{Kuijlaars-Zhang14}, as shown in \cite{Kuijlaars-Stivigny14}.

The macroscopic behavior of the particles as $n\to\infty$ has recently been investigated in \cite{CR14}, where the expressions for the associated equilibrium measures are given for quite general potentials and $\theta\geq 1$. According to \cite{CR14}, as $n\to \infty$, the (rescaled) particles in \eqref{eq:bioLag} are distributed over a finite interval $[0,(1+\theta)^{1+1/\theta}]$, with the density function given by
\begin{equation}
f_\theta(x)= \frac{\theta}{2 \pi x i }(I_+(x)-I_-(x)), \qquad x\in(0,(1+\theta)^{1+1/\theta}).
\end{equation}
Here, $I_{\pm}(x)$ (with $\Im (I_+(x))>0$) stand for two complex conjugate solutions of the equation
$$J(z)=\theta (z+1)\left(\frac{z+1}{z}\right)^{1/\theta}=x, \qquad x\in(0,(1+\theta)^{1+1/\theta}).$$
Moreover, by \cite[Remark 1.9]{CR14}, the density blows up with a rate $x^{-1/(1+\theta)}$ near the origin (hard edge), while vanishes as a square root near $(1+\theta)^{1+1/\theta}$ (soft edge). This phenomenon in particular suggests non-trivial hard edge scaling limits (as shown in \eqref{eq:BorHard}), as well as the expectation that the classical bulk and soft edge universality \cite{Kuijlaars11} (via the sine kernel and Airy kernel, respectively) should hold in the bulk and the right edge as in the case of $\theta=1$. More explicit description is revealed later in \cite{Forrester-Liu14}. After changing variables $x_i \to \theta x_i^{1/\theta}$, the (rescaled) particles are then distributed over $\left[0,(1+\theta)^{1+\theta}/\theta^\theta \right]$ and the limiting mean distribution is recognized as the Fuss-Catalan distribution \cite{Alexeev-Gotze-Tikhomirov10, Banica-Belinschi-Capitaine-Collins11, Nica-Speicher06}. Its $k$-th moment is given by the Fuss-Catalan number
\begin{equation}
\frac{1}{(1+\theta)k+1}\binom{(1+\theta)k+k}{k},\qquad k=0,1,2,\ldots.
\end{equation}
The density function of Fuss-Catalan distribution can be written down explicitly in several ways; cf. \cite{Penson-Zyczkowski11} in terms of Meijer G-functions (see e.g.\ \cite{Beals-Szmigielski13, Luke, DLMF} and
the Appendix below for a brief introduction) or \cite{Liu-Song-Wang11} in terms of multivariate integrals. The simplest form of the representation for general $\theta$ might follow from the following parametrization of the argument \cite{Biane98,Forrester-Liu14,Haagerup-Moller13,Neuschel14}:
\begin{equation}\label{eq:para x}
x=\frac{(\sin ((1+\theta)\varphi))^{1+\theta}}{\sin\varphi (\sin( \theta \varphi))^\theta}, \qquad 0<\varphi<\frac{\pi}{1+\theta}.
\end{equation}
It is readily seen that this parametrization is a strictly decreasing function of $\varphi$, thus gives a  one-to-one mapping from $(0, \pi/(1+\theta))$ to $(0, (1 + \theta)^{1 + \theta}/\theta^\theta)$. The density function in terms of $\varphi$ is then given by
\begin{equation}\label{eq:density}
  \rho(\varphi) = \frac{1}{\pi x} \frac{\sin((1+\theta)\varphi)}{\sin(\theta\varphi)} \sin\varphi = \frac{1}{\pi}\frac{(\sin\varphi)^2(\sin (\theta\varphi))^{\theta-1}}{ (\sin ((1+\theta)\varphi))^\theta},\qquad 0<\varphi<\frac{\pi}{1+\theta}.
\end{equation}
From \eqref{eq:para x} and \eqref{eq:density}, one can check directly that $\rho$ blows up with a rate $x^{-\theta/(1+\theta)}$ near the origin, and vanishes as a square root near $(1+\theta)^{1+\theta}/\theta^\theta$, which is compatible with the changes of variables. We finally note that the other description of macroscopic behavior with the notion of a DT-element \cite{DH04} can be found in \cite{Cheliotis14}.

The main aim of this paper to establish local universality for biorthogonal Laguerre ensembles \eqref{eq:bioLag}. Due to lack of a simple Christoffel-Darboux formula for Laguerre biorthogonal polynomials, we have to adapt an approach that is different from the conventional one. The main issue here is an explicit integral representation of $K_n^{(\alpha,\theta)}$. Our main results are stated in the next section.

\subsection{Statement of the main results}
Our first result is stated as follows:
\begin{thm}[Double contour integral representation of $K_n^{(\alpha,\theta)}$]\label{thm:KnInt}
With $K_n^{(\alpha,\theta)}$ defined in \eqref{eq:kerBio}, we have
\begin{align}\label{eq:Kn_double_con}
     K_n^{(\alpha, \theta)}(x,y) =  \frac{\theta }{(2\pi i)^2} \int_{c-i\infty}^{c+i\infty} \ud s \oint_{\Sigma}  \ud t
   \frac{\Gamma(s+1)\Gamma(\alpha+1+\theta s )}{\Gamma(t+1)\Gamma(\alpha+1+\theta t)}
            \frac{\Gamma(t-n+1)}{\Gamma(s-n+1)}       \frac{x^{-\theta s -1}y^{\theta t}}{s-t},
    \end{align}
for $x,y>0$, where
\begin{equation}\label{eq:defc}
c=\frac{\max\{0,1-\frac{\alpha+1}{\theta}\}-1}{2}<0,
\end{equation} and $\Sigma$ is a closed contour going around $0, 1, \ldots, n-1$ in
the positive direction and  $\Re t > c$ for $t \in \Sigma$.
\end{thm}

We highlight that this contour integral representation bears a resemblance to those appearing recently in the studies of  products of random matrices \cite{Forrester14,KKS15,Kuijlaars-Stivigny14,Kuijlaars-Zhang14}, where the integrands of double contour integral representations for the correlation kernels again consist of ratios of gamma functions. When $\theta\in\mathbb{N}$, $K_n^{(\alpha,\theta)}$ is indeed related to certain correlation kernels arising from products of Ginibre matrices; see Section \ref{sec:integercase} below. We also note that, in the context of products of random matrices, the correlation kernels can be written as integrals involving Meijer G-functions, for biorthogonal Laguerre ensembles, however, it does not seem to be the case for general parameters $\alpha$ and $\theta$.

An immediate consequence of the above theorem is the following new representations of hard edge scaling limits.

\begin{cor}[Hard edge scaling limits of $K_n^{(\alpha,\theta)}$]\label{thm:hardedge}
With $\alpha \geq -1$, $\theta \geq 1$ being fixed, we have
\begin{equation}\label{eq:ZhangHard}
 \lim_{n \to \infty}\frac{K_n^{(\alpha, \theta)}\left(\frac{x}{n^{1/\theta}},\frac{y}{n^{1/\theta}}\right)}{n^{1/\theta}} = K^{(\alpha,\theta)}(x,y),
 \end{equation}
uniformly for $x,y$ in compact subsets of the positive real axis, where
\begin{equation}\label{eq:hardedgelim}
K^{(\alpha,\theta)}(x,y)
=\frac{\theta }{(2\pi i)^2} \int_{c-i\infty}^{c+i\infty} \ud s \oint_{\Sigma}  \ud t
   \frac{\Gamma(s+1)\Gamma(\alpha+1+\theta s )}{\Gamma(t+1)\Gamma(\alpha+1+\theta t)}
            \frac{\sin\pi s}{\sin \pi t}       \frac{x^{-\theta s -1}y^{\theta t }}{s-t}
\end{equation}
and where $c$ is given in \eqref{eq:defc}, $\Sigma$ is a contour starting from $+\infty$ in the upper
half plane and returning to $+\infty$ in the lower half plane which
encircles the positive real axis and $\Re t>c$ for $t\in\Sigma$. Alternatively, by setting
\begin{equation}\label{def:p}
p^{(\alpha,\theta)}(x)=\frac{1}{2 \pi i}\int_{\kappa-i\infty}^{\kappa+i\infty} \frac{\Gamma(\alpha+s)}{\Gamma((1-s)/\theta)}x^{-s} \ud s, \quad \kappa>-\alpha,
\end{equation}
and
\begin{equation}\label{def:q}
q^{(\alpha,\theta)}(x)=\frac{1}{2 \pi i}\int_\gamma \frac{\Gamma(t/\theta)}{\Gamma(\alpha+1-t)}x^{-t} \ud t,
\end{equation}
where $\gamma$ is a loop starting from $-\infty$ in the lower
half plane and returning to $-\infty$ in the upper half plane which
encircles the negative real axis, we have
\begin{equation}\label{eq:hardedgelim2}
K^{(\alpha,\theta)}(x,y)=\int_0^1p^{(\alpha,\theta)}(ux)q^{(\alpha,\theta)}(uy)\ud u.
\end{equation}
\end{cor}

In Corollary \ref{thm:hardedge}, we require $\theta \geq 1$  to make sure the integral is convergent. Note that when $\theta=1$, we have (see \cite[formula 10.9.23]{DLMF})
\begin{align*}
p^{(\alpha,1)}(ux)=(ux)^{\alpha/2}J_{\alpha}(2\sqrt{ux}),\qquad
q^{(\alpha,1)}(uy)=(uy)^{-\alpha/2}J_{\alpha}(2\sqrt{uy}),
\end{align*}
where $J_\alpha$ denotes the Bessel function of the first kind of order $\alpha$. It then follows from \eqref{eq:hardedgelim2} that
\begin{align*}
K^{(\alpha,1)}(x,y) & =\left(\frac{x}{y}\right)^{\alpha/2}\int_0^1
J_{\alpha}(2\sqrt{ux})J_{\alpha}(2\sqrt{uy}) \ud u \\
    & = 4   \left(\frac{x}{y}\right)^{\alpha/2} K^{\rm Bes}_{\alpha}(4x,4y),
    \end{align*}
where
\[ K^{\rm Bes}_{\alpha}(x,y) = \frac{J_\alpha(\sqrt x)\sqrt y J'_\alpha
(\sqrt y)-\sqrt x J_\alpha'(\sqrt x)J_\alpha(\sqrt y)}{2(x-y)},\quad
\alpha>-1,
\]
is the Bessel kernel of order $\alpha$ that appears as the scaling limit of the Laguerre unitary ensembles at the hard edge \cite{Forrester93,TW94}, as expected. Furthermore, a comparison of \eqref{eq:BorHard} and \eqref{eq:ZhangHard}--\eqref{eq:hardedgelim} gives us the following identity
\begin{multline}\label{eq:BorZhang}
\theta x^{\alpha}\int_0^1J_{\frac{\alpha+1}{\theta},\frac{1}{\theta}}(ux)J_{\alpha+1,\theta}((uy)^{\theta})u^\alpha\ud u \\ =
\frac{\theta }{(2\pi i)^2} \int_{c-i\infty}^{c+i\infty} \ud s \oint_{\Sigma}  \ud t
   \frac{\Gamma(s+1)\Gamma(\alpha+1+\theta s )}{\Gamma(t+1)\Gamma(\alpha+1+\theta t)}
            \frac{\sin\pi s}{\sin \pi t}       \frac{x^{-\theta s -1}y^{\theta t }}{s-t},\quad \theta\geq 1.
\end{multline}
For a direct proof of the above formula; see Remark \ref{remark2} below.

We believe that the new integral representations \eqref{eq:hardedgelim} and \eqref{eq:hardedgelim2} for $K^{(\alpha,\theta)}$ will also facilitate further investigations of relevant quantities, say, the differential equations for the associated Fredholm determinants, as done in \cite{Stro14}--\cite{TW94c}. The studies of these aspects will be the topics of future research.

By performing an asymptotic analysis for the double contour integral representation \eqref{eq:Kn_double_con}, we are able to confirm the bulk and soft edge universality for biorthogonal Laguerre ensembles, which are left open in \cite{CR14}. The relevant results are stated as follows.

\begin{thm}[Bulk and soft edge universality]\label{thm:bulk}
For $x_0\in(0, ( 1+\theta)^{1+\theta}/\theta^\theta)$, which is parameterized through \eqref{eq:para x} by  $\varphi = \varphi(x_0) \in (0, \pi/(1+\theta))$, we have, with $\alpha,\theta$ being fixed,
\begin{multline}\label{eq:bulk univ}
 \lim_{n \to \infty} \frac{e^{-\pi\eta \cot\varphi}}{e^{-\pi \xi \cot\varphi}} \frac{1}{\rho(\varphi)x_0^{1-\frac{1}{\theta}}} K_n^{(\alpha,\theta)}\left(n \theta \left(x_0+\frac{\xi}{n\rho(\varphi)}\right)^{\frac{1}{\theta}}, n \theta \left(x_0+\frac{\eta}{n\rho(\varphi)}\right)^{\frac{1}{\theta}}\right) \\
=K_{\sin}(\xi,\eta),
\end{multline}
uniformly for $\xi$ and $\eta$ in any compact subset of $\mathbb{R}$, where $\rho(\varphi)$ is defined in \eqref{eq:density} and
\begin{equation}
K_{\sin}(x,y):=\frac{\sin\pi(x-y)}{\pi(x-y)}
\end{equation}
is the normalized sine kernel.

For the soft edge, we have
  \begin{multline}\label{eq:edge univer}
    \lim_{n \to \infty} \frac{e^{-2^{-\frac{1}{3}}(1+\theta)^{\frac{2}{3}} \eta n^{\frac{1}{3}}}}{e^{-2^{-\frac{1}{3}}(1+\theta)^{\frac{2}{3}} \xi n^{\frac{1}{3}}}}\frac{(1+\theta)^{\frac{2}{3}+\frac{1}{\theta}}}{2^{\frac{1}{3}}} n^{\frac{1}{3}} K_n^{(\alpha,\theta)}\left(n\theta\left( x_* + \frac{ c_* \xi}{n^{\frac{2}{3}}}\right)^{\frac{1}{\theta}}, n\theta \left( x_* + \frac{ c_* \eta}{n^{\frac{2}{3}}}\right)^{\frac{1}{\theta}}\right)
\\ = K_{\Ai}(\xi, \eta)
  \end{multline}
  uniformly for $\xi$ and $\eta$ in any compact subset of $\mathbb{R}$, where
  \begin{equation}\label{def:cast}
    x_* = \frac{( 1+\theta)^{ 1+\theta}}{\theta^\theta}, \quad  \quad c_*=\frac{(1+\theta)^{  \frac{2}{3}+\theta}}{2^{\frac{1}{3}} \theta^{\theta - 1}},
\end{equation}
and
\begin{equation} \label{def:airy kernel}
  K_{\Ai}(x,y):=\frac{\Ai(x)\Ai'(y)-\Ai'(x)\Ai(y)}{x-y} =\frac{1}{(2\pi i)^2}\int_{\gamma_R}\ud \mu\int_{\gamma_L}\ud \lambda \frac{e^{\frac{\mu^3}{3}-x\mu}}{e^{\frac{\lambda^3}{3}-y\lambda}}\frac{1}{\mu-\lambda}
\end{equation}
is the Airy kernel. In \eqref{def:airy kernel}, $\gamma_R$ and $\gamma_L$ are symmetric with respect to the imaginary axis,
and $\gamma_R$ is a contour in the right-half plane going from $e^{-\pi/3i}\cdot \infty$ to $e^{\pi/3 i}\cdot \infty$;
see Figure \ref{fig:Airy_curve} for an illustration.
\end{thm}

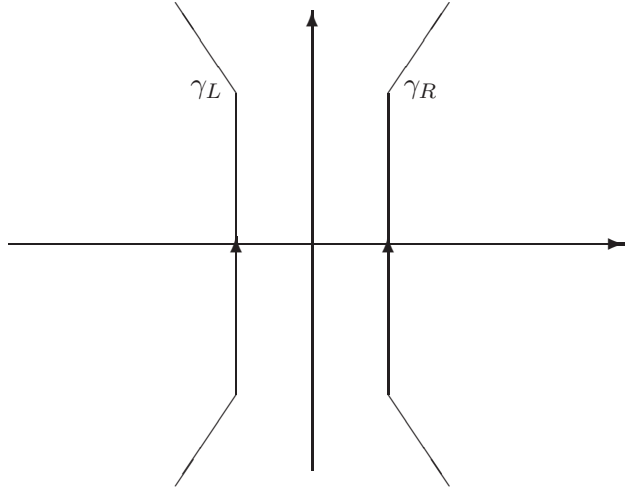
\begin{figure}[t]
\begin{center}
   \setlength{\unitlength}{1truemm}
   \begin{picture}(100,70)(-5,2)
       \put(30,60){\line(0,-1){40}}
       \put(50,60){\line(0,-1){40}}
       \put(50,60){\line(2,3){8}}
       \put(30,60){\line(-2,3){8}}
       \put(50,20){\line(2,-3){8}}
       \put(30,20){\line(-2,-3){8}}
       \put(52,60){$\gamma_R$}
       \put(24,60){$\gamma_L$}
       \put(0,40){\line(1,0){80}}
       \put(40,70){\line(0,-1){60}}

        \put(30,40){\thicklines\vector(0,1){1}}
       \put(50,40){\thicklines\vector(0,1){1}}
       \put(80,40){\thicklines\vector(1,0){1}}
       \put(40,70){\thicklines\vector(0,1){1}}

   \end{picture}

  \caption{ The contours $\gamma_L$ and $\gamma_R$ in the definition of Airy kernel.}
   \label{fig:Airy_curve}
\end{center}
\end{figure}

In the special case $\theta=2$, $\alpha\in\mathbb{N}\cup \{0 \}$, the bulk universality is first proved in \cite{LSZ}.
\begin{remark}
   The result of soft edge universality \eqref{eq:edge univer} also implies that the limiting distribution of the largest particle in biorthogonal Laguerre ensembles, after proper scaling, converges to the well-known Tracy-Widom distribution \cite[Theorem 3.1.5]{Anderson-Guionnet-Zeitouni10}.
\end{remark}

\subsection{Organization of the rest of the paper}

The rest of this paper is organized as follows. Our main results are proved in Section \ref{sec:proofs}. The proofs of Theorem \ref{thm:KnInt} and Corollary \ref{thm:hardedge} are given in Sections \ref{sec:proofThm1} and \ref{sec:proofThm2}, respectively, which rely on  two propositions concerning the contour integral representations of $p_k^{(\alpha,\theta)}$ and $q_k^{(\alpha,\theta)}$ in Section \ref{sec:props}. These formulas might be viewed as extensions of the intensively studied $\theta \in \mathbb{N}$ case, and we give direct proofs here. The nice structures of these formulas then allow us to simplify \eqref{eq:kerBio} into a closed integral form as well as to obtain the hard edge scaling limits, following the idea in recent work of the author with Kuijlaars \cite{Kuijlaars-Zhang14}. The bulk and soft edge universality stated in Theorem \ref{thm:bulk} is proved in Section \ref{sec:proofThm3}. We will perform a steepest descent analysis of the double contour integral \eqref{eq:Kn_double_con}, whose integrand constitutes products and ratios of gamma functions with large arguments. It comes out that the strategy developed by Liu, Wang and the author in \cite{LWZ14} (see also \cite{Adler-van_Moerbeke-Wang11}) works well in the present case. Roughly speaking, the strategy is to approximate the logarithmics of the gamma functions by elementary functions for $n$ large, which play the role of phase functions. There will be two complex conjugate saddle points in the bulk regime, corresponding to the sine kernel, while in the edge regime these two saddle points coalesce into a single one, which leads to the Airy kernel. A crucial feature of the analysis is to construct suitable contours of integration with the aid of the parametrization \eqref{eq:para x}. Since the asymptotic analysis is carried out in a manner similar to that performed in \cite{LWZ14}, emphasis will be placed on key steps and demonstration of basic ideas in the proof of Theorem \ref{thm:bulk}, but refer to \cite{LWZ14} for some technical issues.

We finally focus on the cases when $\theta =M \in\mathbb{N}$, and relate $K_n^{(\alpha,M)}$ to correlation kernels of specified parameters arising from products of $M$ Ginibre matrices. Some remarks are made in accordance with this relation to conclude this paper. For convenience of the reader, we include a short introduction to the Meijer G-function in the Appendix.

\section{Proofs of the main results}\label{sec:proofs}

\subsection{Contour integral representations of $p_k^{(\alpha,\theta)}(x)$ and $q_k^{(\alpha,\theta)}(x)$}\label{sec:props}

\begin{prop}\label{prop:qkint}
We have for $x>0$,
\begin{align}\label{eq:qkint}
q_k^{(\alpha,\theta)}(x) & = (-1)^k\sum_{j=0}^k \binom{k}{j}\frac{(-x)^j}{\Gamma(\alpha+1+j\theta)}\Gamma(\alpha+1+k\theta) \nonumber \\
& = \frac{\Gamma(\alpha+1+k\theta) k!}{2 \pi i}\oint_{\Sigma}\frac{\Gamma(t-k)x^t}{\Gamma(t+1)\Gamma(\alpha+1+\theta t)}  \ud t,
\end{align}
where $\Sigma$ is a closed contour that encircles $0, 1, \ldots, k$ once in the positive direction.
\end{prop}
\begin{proof}
The first identity in \eqref{eq:qkint} follows from the determinantal expressions for the polynomials $q_k^{(\alpha,\theta)}$. By setting the bimoments
\begin{equation}
m_{j,k}=\int_0^{+\infty}x^{\alpha+j+\theta k}e^{-x}\ud x =\Gamma(\alpha+j+k\theta+1), \quad j,k\in\mathbb{N}\cup \{0\},
\end{equation}
we define
\begin{align}\label{def:bimom}
D_n & =\det(m_{j,k})_{j,k=0,\ldots,n} \nonumber \\
&=
\det
\begin{pmatrix}
\Gamma(\alpha+1) & \Gamma(\alpha+1+\theta) & \cdots & \Gamma(\alpha+1+n\theta) \\
\Gamma(\alpha+2) & \Gamma(\alpha+2+\theta) & \cdots &  \Gamma(\alpha+2+n\theta) \\
\vdots &  \vdots & \vdots & \vdots \\
\Gamma(\alpha+1+n) & \Gamma(\alpha+1+n+\theta) & \cdots &\Gamma(\alpha+1+n(\theta+1))
\end{pmatrix}.
\end{align}

From the general theory of biorthogonal polynomials (cf. \cite[Proposition 2]{DF08}), it follows that
\begin{align}\label{eq:qalpha}
q_k^{(\alpha,\theta)}(x)
&=\frac{1}{D_{k-1}}\det
\begin{pmatrix}
m_{0,0} & m_{0,1} & \cdots & m_{0,k} \\
\vdots &  \vdots & \vdots & \vdots \\
m_{k-1,0} & m_{k-1,1} & \cdots & m_{k-1,k} \\
1 & x & \cdots  & x^k
\end{pmatrix} \nonumber
\\
&=\frac{1}{D_{k-1}} \det
\begin{pmatrix}
\Gamma(\alpha+1) & \Gamma(\alpha+1+\theta) & \cdots & \Gamma(\alpha+1+k\theta) \\
\vdots &  \vdots & \vdots & \vdots \\
\Gamma(\alpha+k) & \Gamma(\alpha+k+\theta) & \cdots & \Gamma(\alpha+k(1+\theta)) \\
1 & x & \cdots &  x^k
\end{pmatrix}, ~~k \geq 1
\end{align}
with $q_0^{(\alpha,\theta)}(x)=1$. With the aid of functional relation
\begin{equation}\label{eq:Gamma}
\Gamma(z+1)=z\Gamma(z),
\end{equation}
an easy Gauss elimination process gives us
\begin{equation}\label{eq:detBiMon}
D_n=\prod_{k=0}^{n} k! \theta^k \Gamma(\alpha+1+k\theta).
\end{equation}
Similarly, by expanding the matrix in \eqref{eq:qalpha}
along the last row and evaluating the associated minors, it follows
\begin{equation}\label{eq:qexpli}
q_k^{(\alpha,\theta)}(x)=(-1)^k\sum_{j=0}^k \binom{k}{j}\frac{(-x)^j}{\Gamma(\alpha+1+j\theta)}\Gamma(\alpha+1+k\theta),
\end{equation}
see also \cite{Kon67} for a proof of \eqref{eq:qexpli} by checking the orthogonality directly if $\theta=M$.

To show the second identity in \eqref{eq:qkint}, we note that integrand in the right-hand side of \eqref{eq:qkint} is
meromorphic on $\mathbb C$ with simple poles at $0, 1, \ldots, k$
(the poles of the numerator at the negative integers are canceled by the poles of the factor $\Gamma(t+1)$ in the denominator). Hence, by the residue theorem and a straightforward calculation, we obtain
\begin{align}
&\frac{\Gamma(\alpha+1+k\theta) k!}{2 \pi i}\oint_{\Sigma}\frac{\Gamma(t-k)x^t}{\Gamma(t+1)\Gamma(\alpha+1+\theta t)} \ud t
\nonumber \\
&=\Gamma(\alpha+1+k\theta) k! \sum_{j=0}^k \Res_{t=j}
    \left( \frac{\Gamma(t-k)}{\Gamma(t+1)\Gamma(\alpha+1+\theta t)} \right) x^j
    \nonumber \\
    &=\Gamma(\alpha+1+k\theta) k! \sum_{j=0}^k \frac{(-1)^{k-j}x^j}{(k-j)!j!\Gamma(\alpha+1+j\theta)}
    \nonumber \\
    &=(-1)^k\sum_{j=0}^k \binom{k}{j}\frac{(-x)^j}{\Gamma(\alpha+1+j\theta)}\Gamma(\alpha+1+k\theta).
\end{align}

This completes the proof of Proposition \ref{prop:qkint}.
\end{proof}

\begin{prop}\label{prop:pkint}
For $ p_k^{(\alpha,\theta)}$, we have the following Mellin-Barnes integral representation
\begin{equation} \label{eq:pkint}
    x^\alpha e^{-x} p_k^{(\alpha,\theta)}(x)
    =  \frac{1}{2\pi i \Gamma(\alpha+1+k\theta) k!} \int_{c-i\infty}^{c+i\infty} \frac{\Gamma\left(\frac{s}{\theta}+1-\frac{1}{\theta}\right)}{\Gamma\left(\frac{s}{\theta}+1-\frac{1}{\theta}-k\right)}
    \Gamma(\alpha+s)x^{-s} \ud s,
    \end{equation}
    where $c > \max\{-\alpha, 1-\theta\}$ and $x>0$.
\end{prop}
\begin{proof}
Note that all the poles of the integrand lie on the left of the line $\Re z = c$, it is readily seen that the integral formula in the right-hand side of \eqref{eq:pkint} is well-defined. On account of the uniqueness of biorthogonal functions, our strategy is to check the integral representation satisfies
\begin{itemize}
  \item the orthogonality conditions
  \begin{equation}\label{eq:orthcon}
  \frac{1}{2\pi i \Gamma(\alpha+1+k\theta) k!}\int_0^\infty x^{j\theta} \int_{c-i\infty}^{c+i\infty} \frac{\Gamma\left(\frac{s}{\theta}+1-\frac{1}{\theta}\right)}{\Gamma\left(\frac{s}{\theta}+1-\frac{1}{\theta}-k\right)}
    \Gamma(\alpha+s)x^{-s} \ud s \ud x
    =\delta_{j,k},
  \end{equation}
  for $j=0,1,\ldots,k$;
  \item the integral $\frac{1}{2\pi i } \int_{c-i\infty}^{c+i\infty} \frac{\Gamma\left(\frac{s}{\theta}+1-\frac{1}{\theta}\right)}{\Gamma\left(\frac{s}{\theta}+1-\frac{1}{\theta}-k\right)}
    \Gamma(\alpha+s)x^{-s} \ud s$ belongs to the linear span of $x^\alpha e^{-x},  \newline x^{\alpha+1}e^{-x}, \ldots, x^{\alpha+k}e^{-x}$.
\end{itemize}

To show \eqref{eq:orthcon}, we make use of the inversion formula for the Mellin transform and obtain
\begin{align}
  &\frac{1}{2\pi i \Gamma(\alpha+1+k\theta) k!}\int_0^\infty x^{j\theta} \int_{c-i\infty}^{c+i\infty} \frac{\Gamma\left(\frac{s}{\theta}+1-\frac{1}{\theta}\right)}{\Gamma\left(\frac{s}{\theta}+1-\frac{1}{\theta}-k\right)}
    \Gamma(\alpha+s)x^{-s} \ud s \ud x \nonumber
  \\
  &=\frac{\Gamma\left(\frac{s}{\theta}+1-\frac{1}{\theta}\right)}{\Gamma(\alpha+1+k\theta) k!\Gamma\left(\frac{s}{\theta}+1-\frac{1}{\theta}-k\right)}
    \Gamma(\alpha+s)\bigg{|}_{s=j\theta+1} \nonumber \\
    &=\frac{(j+1-k)_k\Gamma(1+\alpha+j\theta)}{\Gamma(\alpha+1+k\theta) k!}=\delta_{j,k}.
  \end{align}

To check the second statement, recall the Pochhammer symbol $(a)_k=\frac{\Gamma(a+k)}{\Gamma(a)}=a(a+1)\cdots(a+k-1)$, it is readily seen that
$$ \frac{\Gamma\left(\frac{s}{\theta}+1-\frac{1}{\theta}\right)}{\Gamma\left(\frac{s}{\theta}+1-\frac{1}{\theta}-k\right)}=
\left(\frac{s}{\theta}+1-\frac{1}{\theta}-k\right)_k=\left(\frac{s}{\theta}+1-\frac{1}{\theta}-k\right)\cdots
\left(\frac{s}{\theta}+1-\frac{1}{\theta}-1\right)$$
is a polynomials of degree $k$ in $s$, the integral is then a linear combination of weights $w_j^{(\alpha)}(x)$, $j=0,\ldots,k$, where
\begin{align}\label{def:wjalpha}
    w_j^{(\alpha)}(x) &=  \frac{1}{2\pi i} \int_{c-i\infty}^{c+i\infty} s^j \Gamma(\alpha+s) x^{-s} \ud s .
\end{align}
Thus, it suffices to check $w_j^{(\alpha)}(x)$ belongs to the linear span of $x^\alpha e^{-x}, x^{\alpha+1}e^{-x}, \ldots, x^{\alpha+j}e^{-x}$. We now expand the monomial $s^j$ in terms of the basis $(\alpha+s)_l$, $l=0,\ldots,j$, i.e.,
$$ s^j=\sum_{l=0}^j a_l(\alpha+s)_l= \sum_{l=0}^j a_l\frac{\Gamma(\alpha+l+s)}{\Gamma(\alpha+s)}$$
for some constants $a_l$ with $a_j=1$. Inserting the above formula into \eqref{def:wjalpha}, it follows that
\begin{align}
    w_j^{(\alpha)}(x) =  \frac{1}{2\pi i} \sum_{l=0}^j a_l \int_{c-i\infty}^{c+i\infty} \Gamma(\alpha+l+s) x^{-s} \ud s =\sum_{l=0}^j a_l x^{\alpha+l}e^{-x},
\end{align}
as desired, where we have made use of the fact that
$$\frac{1}{2\pi i} \int_{c-i\infty}^{c+i\infty} \Gamma(\nu+s) x^{-s} \ud s=x^\nu e^{-x}, \qquad \nu>-1;$$
see \eqref{eq:LaginMei} below.

This completes the proof of Proposition \ref{prop:pkint}.
\end{proof}

\subsection{Proof of Theorem \ref{thm:KnInt}}\label{sec:proofThm1}
With a change of variable $s\rightarrow \theta s+1-\theta$ in \eqref{eq:pkint} and contour deformation, we rewrite $x^\alpha e^{-x} p_k^{(\alpha,\theta)}(x)$ as
\begin{equation}\label{eq:pkint2}
\frac{\theta x^{\theta-1}}{2\pi i \Gamma(\alpha+1+k\theta) k! } \int_{c-i\infty}^{c+i\infty} \frac{\Gamma\left(s\right)}{\Gamma\left(s-k\right)}
    \Gamma(\theta s +1-\theta+\alpha)x^{-\theta s } \ud s,
    \end{equation}
where $c>\max\{0,1-\frac{\alpha+1}{\theta}\}$. This, together with \eqref{eq:kerBio} and \eqref{eq:qkint}, implies that
\begin{equation} \label{Knintegral0}
    K_n^{(\alpha,\theta)}(x,y) = \frac{\theta x^{\theta-1}}{(2\pi i)^2} \int_{c-i\infty}^{c+i\infty} \ud s \oint_{\Sigma}  \ud t
    \frac{\Gamma(s)\Gamma( \theta s +1-\theta+\alpha)}{\Gamma(t+1)\Gamma(\alpha+1+\theta t)}
\sum_{k=0}^{n-1} \frac{\Gamma(t-k)}{\Gamma(s-k)} x^{-\theta s }y^{\theta t} .
    \end{equation}
We now follow the idea in \cite{Kuijlaars-Zhang14}. From the functional equation \eqref{eq:Gamma}, one can
easily check that
\[ (s-t-1) \frac{\Gamma(t-k)}{\Gamma(s-k)} =
\frac{\Gamma(t-k)}{\Gamma(s-k-1)} -
\frac{\Gamma(t-k+1)}{\Gamma(s-k)}, \] which means that there is a
telescoping sum
\begin{equation} \label{Gammatelescope}
    (s-t-1) \sum_{k=0}^{n-1} \frac{\Gamma(t-k)}{\Gamma(s-k)}
    = \frac{\Gamma(t-n+1)}{\Gamma(s-n)} -\frac{\Gamma(t+1)}{\Gamma(s)}.
        \end{equation}

To make sure that $s-t-1 \neq 0$ when $s \in c + i\mathbb R$ and $t \in \Sigma$, we make the following settings. Note that $\max\{0,1-\frac{\alpha+1}{\theta}\}<1$ for $\alpha\geq -1$ and $\theta>0$, we take
$$c=\frac{1+\max\{0,1-\frac{\alpha+1}{\theta}\}}{2}<1$$ and let $\Sigma$ go around $0, 1, \ldots, n-1$ but with $\Re t > c-1$ for $t \in \Sigma$. Then we insert \eqref{Gammatelescope} into
\eqref{Knintegral0} and get
\begin{align*}
    K_n^{(\alpha, \theta)}(x,y) = &  \frac{\theta x^{\theta-1}}{(2\pi i)^2} \int_{c-i\infty}^{c+i\infty} \ud s \oint_{\Sigma}  \ud t
   \frac{\Gamma(s)\Gamma(\theta s +1-\theta+\alpha)}{\Gamma(t+1)\Gamma(\alpha+1+\theta t)}
            \frac{\Gamma(t-n+1)}{\Gamma(s-n)}       \frac{x^{-\theta s}y^{\theta t }}{s-t-1}
\nonumber \\ &- \frac{\theta x^{\theta-1}}{(2\pi i)^2}
\int_{c-i\infty}^{c+i\infty} \ud s \oint_{\Sigma}  \ud t
    \frac{\Gamma(\theta s +1-\theta+\alpha)}{\Gamma(\alpha+1+\theta t)}
            \frac{x^{-\theta s }y^{\theta t }}{s-t-1}.
    \end{align*}
The $t$-integral in the second double integral vanishes due to Cauchy's
theorem, since there are no singularities for the integrand inside
$\Sigma$. With a change of variable $s \mapsto s+1$ in the first double integral, we obtain
\eqref{eq:Kn_double_con}

This completes the proof of Theorem \ref{thm:KnInt}.

\subsection{Proof of Corollary \ref{thm:hardedge}}\label{sec:proofThm2}
The proof now is straightforward by taking limit in \eqref{eq:Kn_double_con}, as in \cite{Kuijlaars-Zhang14}.
Recall the reflection formula of the gamma function
\begin{equation}\label{eq:reflection}
\Gamma(t) \Gamma(1-t) = \frac{\pi}{\sin \pi t},
\end{equation}
it is readily seen that
\begin{equation}\label{eq:n to -n}
\frac{\Gamma(t-n+1)}{\Gamma(s-n+1)} =\frac{\Gamma(n-s)}{\Gamma(n-t)}
\frac{\sin \pi s}{\sin \pi t}.
\end{equation}

As $n \to \infty$, we have (cf. \cite[formula 5.11.13]{DLMF})
\begin{equation}\label{eq:ratio asy}
\frac{\Gamma(n-s)}{\Gamma(n-t)} = n^{t-s} \left(1+ O(n^{-1})\right),
\end{equation}
which can be easily verified by using Stirling's formula for the gamma functions.
By modifying the contour $\Sigma$ in \eqref{eq:Kn_double_con} from a
closed contour around $0, 1, \ldots, n-1$ to a two sided unbounded
contour starting from $+\infty$ in the upper
half plane and returning to $+\infty$ in the lower half plane which
encircles the positive real axis and $\Re t>c$ for $t\in\Sigma$, the scaling limits \eqref{eq:hardedgelim} follow. The interchange of limit and integrals can be justified by combining elementary estimates of the $\sin$ and gamma functions with the dominated convergence theorem, as explained in \cite{Kuijlaars-Zhang14}.

To show \eqref{eq:hardedgelim2}, we note that
\begin{equation} \label{u-integral}
    \frac{x^{-\theta s -1}y^{t \theta }}{s-t} = - \theta \int_0^1 (ux)^{-\theta s -1} (uy)^{\theta t } \ud u,
    \end{equation}
and, by \eqref{eq:reflection}, \[\frac{\sin \pi s}{\sin \pi t} =
\frac{\Gamma(1+t)\Gamma(-t)}{\Gamma(1+s) \Gamma(-s)}.\] Inserting the above two formulas into \eqref{eq:hardedgelim}, it is readily seen that
\begin{multline*}
    K^{(\alpha,\theta)}(x,y) = -  \int_0^1 \left( \frac{\theta}{2\pi i} \int_{c-i\infty}^{c+i\infty} \frac{\Gamma(\alpha+1+\theta s)}{\Gamma(-s)} (ux)^{-\theta s-1} \ud s \right) \\
    \times \left( \frac{\theta}{2\pi i} \int_{\Sigma} \frac{\Gamma(-t)}{\Gamma(\alpha+1+\theta t)} (uy)^{\theta t} \ud t \right)
    \ud u.
\end{multline*}
The change of variables $s \mapsto \theta s+1$ and $t \mapsto -\theta t$ takes the two integrals into the two functions $p^{(\alpha,\theta)}$ and $q^{(\alpha,\theta)}$ defined in \eqref{def:p} and \eqref{def:q}, respectively. The identity \eqref{eq:hardedgelim2} then follows.

This completes the proof of Theorem \ref{thm:hardedge}.

\subsection{Proof of Theorem \ref{thm:bulk}}\label{sec:proofThm3}

We start with a scaling of the correlation kernel $K_n^{(\alpha,\theta)}(x,y) \to K_n^{(\alpha,\theta)}(\theta x^{\frac{1}{\theta}},\theta y^{\frac{1}{\theta}})$.  By \eqref{eq:Kn_double_con}, it then follows that
\begin{multline}\label{eq:KnScaling}
  K_n^{(\alpha, \theta)}(\theta x^{\frac{1}{\theta}},\theta y^{\frac{1}{\theta}})
\\=  \frac{1}{(2\pi i)^2 x^{\frac{1}{\theta}}} \int_{\mathcal{C}} \ud s \oint_{\Sigma}  \ud t
   \frac{\Gamma(s+1)\Gamma(\alpha+1+\theta s )}{\Gamma(t+1)\Gamma(\alpha+1+\theta t)}
            \frac{\Gamma(t-n+1)}{\Gamma(s-n+1)}       \frac{\theta^{-\theta s} x^{-s}\theta^{\theta t} y^{t}}{s-t},
\end{multline}
where $\mathcal{C}$ and $\Sigma$ are two contours to be specified later, depending on the choices of reference points.

By setting
\begin{equation} \label{eq:defn_F}
  F(z; a) := \log \left( \frac{\Gamma(z+1)\Gamma(\alpha+1+ \theta z)}{\Gamma(z-n+1)}
              \theta^{-\theta z} a^{-z} \right), \qquad a\geq 0,
\end{equation}
where the branch cut for the logarithmic function is taken along the negative axis and we assume that the value of $\log z$ for $z \in (-\infty, 0)$ is continued from above, we could rewrite \eqref{eq:KnScaling} as
\begin{equation}\label{eq:KinF}
  K_n^{(\alpha, \theta)}(\theta x^{\frac{1}{\theta}},\theta y^{\frac{1}{\theta}})
=  \frac{1}{(2\pi i)^2 x^{\frac{1}{\theta}}} \int_{\mathcal{C}} \ud s \oint_{\Sigma}  \ud t \frac{e^{F(s; x)}}{e^{F(t; y)}} \frac{1}{s - t}.
\end{equation}

We will then perform an asymptotic analysis of \eqref{eq:KinF}. The basic idea is the following. It is clear that the function $F$ in \eqref{eq:KinF} plays the role of a phase function. For large $z$ and proper scalings, $F$ can be approximated by a more elementary function $\hat{F}$ (see \eqref{eq:hat_F} below) with the help of the Stirling's formula for gamma function. There will be two complex conjugate saddle points $w_{\pm}$ (see \eqref{eq:defn_w_pm} below) of $\hat F$ in general. In the proof of bulk universality, the two contours are deformed so that one of them will meet the pair of saddle points. It comes out that the main contribution to the integral does not come from the saddle points alone, but from the vertical line segment connecting the two points. In the proof of soft edge universality, the two saddle points coalesce into a real one. The phase function then behaves like a cubic polynomial around the saddle point (see \eqref{eq:F_hat_Taylor_expansion} below), which justifies the appearance of Airy kernel.

We also note the possibilities to deform the contours in \eqref{eq:KinF}. Firstly, it is readily seen that the integral contour for $s$ can be replaced by any infinite contour $\mathcal{C}$ oriented from $-i\infty$ to $i\infty$, as long as $\Sigma$ is on the right side of $\mathcal{C}$. One can further deform $\mathcal{C}$ such that $\Sigma$ is on its left, and the resulting double contour integral remains the same. To see this, let $\mathcal{C}$ and $\mathcal{C}'$ be two infinite contours from $-i\infty$ to $i\infty$ such that $\Sigma$ lies between $\mathcal{C}$ and $\mathcal{C}'$. An appeal to the residue theorem to the integral on $\mathcal{C} \cup \mathcal{C}'$ gives
\begin{equation}
  \int_\mathcal{C} \ud s \oint_{\Sigma}  \ud t \frac{e^{F(s; x)}}{e^{F(t; y)}} \frac{1}{s - t} - \int_{\mathcal{C}'} \ud s \oint_{\Sigma}  \ud t \frac{e^{F(s; x)}}{e^{F(t; y)}} \frac{1}{s - t} = 2\pi i \int_{\Sigma} \left( \frac{y}{x} \right)^t \ud t = 0.
\end{equation}
Hence, the double contour integral does not change if $\mathcal{C}$ is replaced by $\mathcal{C}'$. We will use such kind of contour deformation in the proof of the soft edge universality. Similarly, one can show that if $\Sigma$ is split into two disjoint closed counterclockwise contours $\Sigma = \Sigma_1 \cup \Sigma_2$, which jointly enclose poles $0, 1, \ldots, n-1$, and $\mathcal{C}$ is an infinite contour from $-i\infty$ to $i\infty$ such that $\Sigma_1$ is on the left side of $\mathcal{C}$ and $\Sigma_2$ is on the right side of $\mathcal{C}$, the formula \eqref{eq:KinF} is still valid. We will use such kind of contours in the proof of the bulk universality.

We now derive the asymptotic behavior of $F$. Recall that the Stirling's formula for gamma function \cite[formula 5.11.1]{DLMF} reads
\begin{equation}\label{eq:stirling}
\log \Gamma(z) = \left(z-\frac{1}{2}\right)\log z-z+\frac{1}{2}\log (2\pi)+\mathcal{O}\left(\frac{1}{z}\right)
\end{equation}
as $z\to\infty$ in the sector $\lvert \arg z \rvert \leq \pi-\epsilon $ for some $\epsilon > 0$. It then follows that if $\lvert z \rvert \to \infty$ and $\lvert z - n \rvert \to \infty$, while $\arg z$ and $\arg(z - n)$ are in $(-\pi + \epsilon, \pi - \epsilon)$, then uniformly
\begin{equation} \label{eq:F_in_F_tilde}
  F(z; a) = \tilde{F}(z; a) + \frac{1}{2}( \log z - \log(z - n))+ \frac{1}{2} \log(2\pi) +\bigO(\min(\lvert z \rvert, \lvert z - n \rvert)^{-1}),
\end{equation}
where
\begin{equation} \label{eq:tilde_F_intermidiate}
  \tilde{F}(z; a) = (1+\theta)z(\log z - 1) - (z - n)(\log(z - n) - 1) - z\log a.
\end{equation}
Furthermore, we have
\begin{equation} \label{eq:defn_hat_F}
  \tilde{F}(nz; n^\theta a) = n \hat{F}(z; a) + n\log n,
\end{equation}
where
\begin{equation} \label{eq:hat_F}
  \hat{F}(z; a) = (1+\theta)z(\log z - 1) - (z - 1)(\log(z - 1) - 1) - z\log a.
\end{equation}
Note that if $\theta=M\in\mathbb{N}$, we encounter the same $\tilde{F}(z; a)$ and $\hat{F}(z; a)$ as in \cite{LWZ14}.

Since
\begin{equation}
\hat{F}_{z}(z;x)=(1+\theta)\log z-\log(z-1)-\log x,
\end{equation}
the saddle point of $\hat{F}(z;x)$ satisfies the equation
\begin{equation}\label{eq:algequ}
z^{1+\theta}=x(z-1).
\end{equation}
In particular, if $x=x_0\in(0, (1 + \theta)^{1 + \theta}/\theta^\theta)$, which is parameterized through \eqref{eq:para x} by  $\varphi = \varphi(x_0) \in (0, \pi/(1+\theta))$, one can find two complex conjugate solutions of \eqref{eq:algequ} explicitly given by
\begin{equation} \label{eq:defn_w_pm}
  w_{\pm} = \frac{\sin((1+\theta)\varphi)}{\sin(\theta\varphi)} e^{\pm i\varphi}.
\end{equation}

For later use, we also define a closed contour
\begin{equation}
\tilde{\Sigma}=\left\{
z=\frac{ \sin ( (1+\theta) \phi) } { \sin( \theta\phi)  } \,  e^{i\phi}
~\Big{|}~ -\frac{\pi}{1+\theta} \leq \phi\leq \frac{\pi}{1+\theta}
\right\},
\end{equation}
which passes through $w_\pm$, intersects the real line only at $0$ when $\phi =  \pm \pi/(1 + \theta)$ and at $1 + \theta^{-1}$ when $\phi = 0$. Since the integrand of \eqref{eq:KinF} takes $0$ as one of the poles, we further deform $\tilde{\Sigma}$ a little bit near the origin simply by  setting
\begin{equation}
  \tilde{\Sigma}^{\epsilon} := \text{$\{ z \in \tilde{\Sigma} \mid \lvert z \rvert \geq \epsilon \}$} \cup \text{the arc of $\{ \lvert z \rvert = \epsilon \}$ connecting $\tilde{\Sigma} \cap \{ \lvert z \rvert = \epsilon \}$ and through $-\epsilon$},
\end{equation}
with counterclockwise orientation.

With the above preparations, we are ready to prove the bulk and soft edge universality for $K_n^{(\alpha,\theta)}$.

\paragraph{Proof of \eqref{eq:bulk univ}}
In view of \eqref{eq:bulk univ}, we scale the arguments $x$ and $y$ in \eqref{eq:KinF} such that
\begin{equation} \label{eq:defn_xy_bulk}
x = n^\theta \left(x_0 + \frac{ \xi}{n \rho(\varphi)}\right), \qquad y =n^\theta \left(x_0 + \frac{\eta}{n \rho(\varphi)}\right),
\end{equation}
where $\xi$ and $\eta$ are in a compact subset of $\realR$ and $\rho(\varphi)$ is given in \eqref{eq:density}.

The contours $\mathcal{C}$ and $\Sigma$ are chosen in the following ways. The contour $\mathcal{C}$ is simply taken to be an upward straight line passing through two scaled saddle points $n w_{\pm}$. This line then divides $n \tilde{\Sigma}^{r}$ into two parts, where $r$ is a small parameter depending on $\theta$. By further separating these two parts, we define
\begin{equation} \label{eq:Sigma_contour_bulk}
  \Sigma = \Sigma_{\curved} \cup \Sigma_{\vertical},
\end{equation}
where $\Sigma_{\curved}$ is the part from $n \tilde{\Sigma}^{r}$, and $\Sigma_{\vertical}$ are two vertical lines connecting ending points of $\Sigma_{\curved}$. The distance of these two lines is taken to be $2\epsilon$, with $\mathcal{C}$ lying in the middle of them; see Figure \ref{fig:Bulk} for an illustration. The main issue here is that, with these choices of $\mathcal{C}$ and $\Sigma$, $\Re \hat{F}(z; x_0)$ defined in \eqref{eq:hat_F} attains its global maximum at $z=w_{\pm}$ for $nz \in \mathcal{C}$ and its global minimum at $z = w_{\pm}$ for $z \in \tilde \Sigma$, which can be proved rigorously with estimates as shown in \cite[Lemma 3.1]{LWZ14}.
\begin{figure}[t]
\begin{center}
   \setlength{\unitlength}{1truemm}
   \begin{picture}(100,70)(-5,2)

       \qbezier(0,40)(0,45)(5,45)
       \qbezier(0,40)(0,35)(5,35)
       \qbezier(5,45)(8,59.5)(35,61)
       \qbezier(35,61)(40,60)(40,59)
       \qbezier(50,50)(60.5,40)(50,30)

       \qbezier(5,35)(8,21.5)(35,19)
       \qbezier(35,19)(40,20)(40,21)

        \put(40,21){\line(0,1){38}}

        \put(50,30){\line(0,1){20}}

        \put(5,40){\thicklines\circle*{0.5}}

        \put(50,40){\thicklines\vector(0,1){1}}

        \put(40,40){\thicklines\vector(0,1){1}}

        \put(45,40){\thicklines\vector(0,1){1}}

        \put(45,16){\line(0,1){48}}

        \put(46,64){$\mathcal{C}$}

        \put(51,40){$\Sigma_{\vertical}$}

        \put(51,50){$\Sigma_{\curved}$}

         \put(32,40){$\Sigma_{\vertical}$}

         \put(32,62){$\Sigma_{\curved}$}

\end{picture}
   \vspace{-17mm}
   \caption{The contours $\mathcal{C}$ and $\Sigma$ used in the proof of bulk universality.}
  \label{fig:Bulk}
\end{center}
\end{figure}
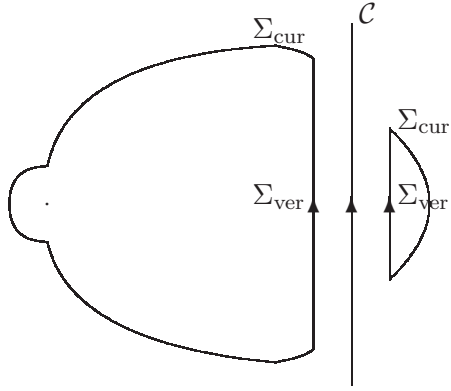

By taking the limit $\epsilon \to 0$, it follows
\begin{equation}\label{eq:spliKn}
  K_n^{(\alpha, \theta)}(\theta x^{\frac{1}{\theta}},\theta y^{\frac{1}{\theta}}) = I_1 + I_2,
\end{equation}
where ($\pv$ means the Cauchy principal value)
\begin{multline} \label{eq:formula_I_1}
  I_1 :=  \lim_{\epsilon \to 0} \frac{1}{(2\pi i)^2 x^{\frac{1}{\theta}}} \int_\mathcal{C} \ud s \int_{\Sigma_{\curved}} \ud t \frac{e^{F(s; x)}}{e^{F(t; y)}} \frac{1}{s - t} \\
  =  \frac{1}{(2\pi i)^2 x^{\frac{1}{\theta}}} \pv \int_{n\tilde{\Sigma}^r} \left( \int_\mathcal{C} \ud s \frac{e^{F(s; x)}}{e^{F(t; y)}} \frac{1}{s - t} \right)\ud t ,
  \end{multline}
  and, by interchange of integrals and the Cauchy's theorem,
  \begin{multline} \label{eq:formula_I_2}
      I_2 := \lim_{\epsilon \to 0} \frac{1}{(2\pi i)^2 x^{\frac{1}{\theta}}} \int_\mathcal{C} \ud s \int_{\Sigma_{\vertical}} \ud t \frac{e^{F(s; x)}}{e^{F(t; y)}} \frac{1}{s - t}
      = \frac{1}{2\pi i x^{\frac{1}{\theta}} } \int^{nw_+}_{nw_-} \frac{e^{F(s; x)}}{e^{F(s; y)}}\ud s \\
= \frac{1}{2\pi i  x^{\frac{1}{\theta}}} \int^{nw_+}_{nw_-} \left( \frac{y}{x} \right)^s \ud s
      =  \frac{1}{2\pi i x^{\frac{1}{\theta}} \log(\frac{y}{x})} \left( \left( \frac{y}{x} \right)^{nw_+} - \left( \frac{y}{x} \right)^{nw_-} \right).
\end{multline}
Here we note that by taking $\epsilon \to 0$, the vertical line $(nw_-,nw_+)$ is enclosed by $\Sigma_{\vertical}$, hence the Cauchy's theorem is applicable in the first step.

With the values of $x, y$ given in \eqref{eq:defn_xy_bulk} and $w_\pm$ given in \eqref{eq:defn_w_pm}, a straightforward calculation gives us
\begin{align} \label{eq:result_I_1}
 I_2 = {}& \frac{\rho(\varphi)x_0^{1-\frac{1}{\theta}}}{2\pi i (\eta - \xi)\left(1 + \bigO \left(n^{-1}\right)\right)} \left( e^{\frac{(\eta - \xi)w_+}{\rho(\varphi)x_0}} \left(1 + \bigO \left(n^{-1}\right)\right) - e^{\frac{(\eta - \xi)w_-}{\rho(\varphi)x_0}} \left(1 + \bigO \left(n^{-1}\right)\right) \right)  \nonumber
 \\
 = {}& \rho(\varphi) x_0^{1-\frac{1}{\theta}} \frac{e^{\pi \cot\varphi \eta}}{e^{\pi \cot\varphi \xi}}  \frac{\sin\pi(\xi - \eta)}{\pi(\xi - \eta)} + \bigO\left(n^{-1}\right)
 \end{align}

On the other hand, one can show that, in a manner similar to the estimates in \cite[Lemma 2.1]{LWZ14}, $F(z; n^\theta x_0)$  attains its global maximum at  $z=nw_{\pm}$ for $z \in \mathcal{C}$ and its global minimum at $z = n w_{\pm}$ for $z \in \tilde \Sigma$, which leads to the fact that $I_1(z)=\mathcal{O}(n^{-1/2})$. This, together with \eqref{eq:spliKn} and \eqref{eq:result_I_1}, implies \eqref{eq:bulk univ}.

\paragraph{Proof of \eqref{eq:edge univer}}
On account of the scalings of $x,y$ in \eqref{eq:edge univer}, we set
\begin{equation} \label{eq:scaling_of_xy_Airy}
x = n^\theta \left(x_\ast + \frac{c_\ast \xi}{n^{2/3}}\right), \qquad y = n^\theta \left(x_\ast + \frac{c_\ast \eta}{n^{2/3}}\right),
\end{equation}
where $\xi,\eta \in\mathbb{R}$, $x_\ast$ and $c_\ast$ are given in \eqref{def:cast}.

In this case, the two saddle points $w_\pm$ coalesce into a single one, i.e.,
\begin{equation}\label{def:z0}
w_+=w_-=z_0=1+\frac{1}{\theta}.
\end{equation}
We now select the contours $\Sigma$ and $\mathcal{C}$ as illustrated in Figure \ref{fig:SoftContour}.
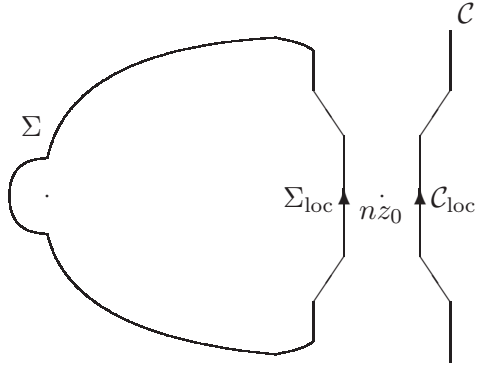
\begin{figure}[t]
\begin{center}
   \setlength{\unitlength}{1truemm}
   \begin{picture}(100,70)(-5,2)

       \qbezier(0,40)(0,45)(5,45)
       \qbezier(0,40)(0,35)(5,35)
       \qbezier(5,45)(8,59.5)(35,61)
       \qbezier(35,61)(40,60)(40,59)

       \qbezier(5,35)(8,21.5)(35,19)
       \qbezier(35,19)(40,20)(40,21)

        \put(40,21){\line(0,1){5}}

        \put(40,26){\line(2,3){4}}

        \put(44,32){\line(0,1){16}}

        \put(40,59){\line(0,-1){5}}

        \put(40,54){\line(2,-3){4}}

        \put(58,18){\line(0,1){8}}

        \put(58,26){\line(-2,3){4}}

        \put(54,32){\line(0,1){16}}

        \put(58,62){\line(0,-1){8}}

        \put(58,54){\line(-2,-3){4}}

        \put(5,40){\thicklines\circle*{0.5}}

        \put(49,40){\thicklines\circle*{0.5}}

        \put(44,40){\thicklines\vector(0,1){1}}

        \put(54,40){\thicklines\vector(0,1){1}}


        \put(59,63){$\mathcal{C}$}

        \put(1.5,48){$\Sigma$}

         \put(36,38.5){$\Sigma_{\loc}$}

         \put(55.5,38.5){$\mathcal{C}_{\loc}$}

         \put(46,37){$nz_0$}
   \end{picture}
   \vspace{-19mm}
   \caption{The contours $\mathcal{C}$ and $\Sigma$ used in the proof of soft edge universality.}
  \label{fig:SoftContour}
\end{center}
\end{figure}
The contour $\Sigma$ is still a deformation of $n\tilde{\Sigma}^{r}$, while near the scaled saddle point $nz_0$, the local part $\Sigma_{\local}$ is defined by
\begin{multline}\label{eq:sigmlocal}
    \Sigma_{\local} = \left\{ nz_0 + c_1 n^{\frac{2}{3}} re^{2\pi i/3} ~\Big{|}~ r \in \left[1, n^{\frac{1}{30}}\right] \right\} \cup \left\{ nz_0 + c_1 n^{\frac{2}{3}} re^{-2\pi i/3} ~\Big{|}~ r \in \left[1, n^{\frac{1}{30}}\right] \right\} \\
    \cup \left\{ nz_0 - \frac{c_1 n^{\frac{2}{3}}}{2} + ic_1 n^{\frac{2}{3}} r ~\Big{|}~ r \in \left [-\frac{\sqrt{3}}{2}, \frac{\sqrt{3}}{2}\right] \right\},
\end{multline}
where
\begin{equation}\label{def:c1}
c_1:=x_\ast/c_\ast=2^{\frac{1}{3}}(1+\theta)^{\frac{1}{3}}/\theta.
\end{equation}
The contour $\mathcal{C}$ is obtained by deforming a straight line. Around $nz_0$, the local part is defined by
\begin{multline}\label{eq:clocal}
  \mathcal{C}_{\local} = \left\{ nz_0 + c_1 n^{\frac{2}{3}} re^{\pi i/3} ~\Big{|}~ r \in \left[1, n^{\frac{1}{30}}\right] \right\} \cup \left\{ nz_0 + c_1 n^{\frac{2}{3}} re^{-\pi i/3} ~\Big{|}~ r \in \left[1, n^{\frac{1}{30}}\right] \right\} \\
  \cup \left\{ nz_0 + \frac{c_1 n^{\frac{2}{3}}}{2} + ic_1 n^{\frac{2}{3}} r ~\Big{|}~ r \in \left[-\frac{\sqrt{3}}{2}, \frac{\sqrt{3}}{2}\right] \right\}.
\end{multline}

As in \cite[Equation (2.69)]{LWZ14}, one can show that the main contribution to the integral \eqref{eq:KinF}, as $n\to\infty$, comes from the part $\mathcal{C}_{\loc} \times \Sigma_{\loc}$, and the remaining part of the integral is negligible. When $(s,t)\in \mathcal{C}_{\loc} \times \Sigma_{\loc}$, we can approximate $F(s; n^\theta x_\ast)$ and $F(t; n^\theta x_\ast)$ by $\tilde{F}$ given in \eqref{eq:F_in_F_tilde}, and further by $\hat{F}$ that is defined in \eqref{eq:hat_F}.

With $z_0$ given in \eqref{def:z0}, it is readily seen that
\begin{equation}
  \hat{F}_z(z_0; x_\ast) = 0, \quad \hat{F}_{zz}(z_0; x_\ast) = 0, \quad \hat{F}_{zzz}(z_0; x_\ast) = \frac{\theta^3}{1+\theta}.
\end{equation}
Hence,
\begin{align} \label{eq:F_hat_Taylor_expansion}
    & \hat{F}(z_0 + n^{-\frac{1}{3}}c_1u; x_\ast )  \nonumber \\
    = {}& \hat{F}(z_0; x_\ast) + \hat{F}_z(z_0; x_\ast)c_1u n^{-\frac{1}{3}} + \frac{1}{2} \hat{F}_{zz}(z_0; x_\ast)c^2_1u^2 n^{-\frac{2}{3}} + \frac{1}{6} \hat{F}_{zzz}(z_0; x_\ast)c^3_1u^3 n^{-1} + \bigO\left(n^{-\frac{6}{5}}\right)  \nonumber \\
    = {}& \hat{F}(z_0; x_\ast) + \frac{u^3}{3n} + \bigO\left(n^{-\frac{6}{5}}\right).
\end{align}
By changes of variables
\begin{equation} \label{eq:parametrization_Airy}
  s = nz_0 + n^{\frac{2}{3}} c_1 u, \qquad t = nz_0 + n^{\frac{2}{3}} c_1 v,
\end{equation}
it follows from \eqref{eq:KinF}, \eqref{eq:F_in_F_tilde}, \eqref{eq:scaling_of_xy_Airy}, \eqref{def:c1} and \eqref{eq:F_hat_Taylor_expansion} that
\begin{align} \label{eq:local_asy_analysis_Airy}
  & \frac{1}{(2\pi i)^2 x^{\frac{1}{\theta}}} \int_{\mathcal{C}_\loc} \ud s \oint_{\Sigma_\loc}  \ud t \frac{e^{F(s; x)}}{e^{F(t; y)}} \frac{1}{s - t} \nonumber \\
={}& \frac{1}{(2\pi i)^2 x^{\frac{1}{\theta}}} \int_{\mathcal{C}_{\loc}} \ud s \oint_{\Sigma_{\loc}}  \ud t \frac{e^{F(s; n^M x_\ast)}\left (1 + n^{-\frac{2}{3}} c^{-1}_1 \xi \right)^{-s}}{e^{F(t; n^M x_\ast)} \left(1 + n^{-\frac{2}{3}} c^{-1}_1 \eta \right)^{-t}} \frac{1}{s - t} \nonumber \\
   = {}& \frac{e^{2^{-\frac{1}{3}}(1+\theta)^{\frac{2}{3}} (\eta - \xi) n^{\frac{1}{3}}}c_1}{n^{\frac{1}{3}} x_\ast^{\frac{1}{\theta}} } \left( \frac{1}{(2\pi i)^2} \int_{\mathcal{C}_{0}} \ud u \int_{\Sigma_{0}} \ud v \frac{e^{ \frac{1}{3}u^3 - u\xi}}{e^{\frac{1}{3}v^3 - v\eta}} \frac{1}{u - v} + \bigO\left(n^{-\frac{1}{5}}\right) \right) \nonumber \\
    = {}& \frac{2^{\frac{1}{3}}e^{2^{-\frac{1}{3}}( 1 + \theta)^{\frac{2}{3}} (\eta - \xi) n^{\frac{1}{3}}}}{n^{ \frac{1}{3}}(1+\theta)^{\frac{1}{\theta}+\frac{2}{3}}} \left( K_{\Ai}(\xi, \eta) + \bigO\left(n^{-\frac{1}{5}}\right) \right),
 \end{align}
where  $\Sigma_0$ and $\mathcal{C}_0$ are the images of $\mathcal{C}_{\loc}$ and $\Sigma_{\loc}$ (see \eqref{eq:clocal} and \eqref{eq:sigmlocal}) under the change of variables \eqref{eq:parametrization_Airy}, and the last equality follows from the integral representation of Airy kernel shown in \eqref{def:airy kernel}.

This completes the proof of Theorem \ref{thm:bulk}.

\section{The cases when $\theta=M\in\mathbb{N}$}\label{sec:integercase}
In this section, we will show a remarkable connection between $K^{(\alpha,\theta)}_n$ and those arising from products of Ginibre random matrices if $\theta\in\mathbb{N}$. In the limiting case, this relation has been established in \cite{Kuijlaars-Stivigny14}. Our result gives new insights for the relations between these two different determinantal point processes. In particular, it provides the other perspective to explain the appearance of Fuss-Catalan distribution in biorthogonal Laguerre ensembles; see Remark \ref{remark1} below. We start with an introduction to the correlation kernels appearing in recent investigations of products of Ginibre matrices.

\subsection{Correlation kernels arising from products of $M$ Ginibre matrices}

Let $X_j$, $j=1,\ldots,M$ be independent complex matrices of size $(n+\nu_j)\times(n+\nu_{j-1})$ with $\nu_0=0$ and $\nu_j\geq 0$. Each matrix has independent and identically distributed standard complex Gaussian entries.
These matrices are also known as Ginibre random matrices. We then form the product
\begin{equation} \label{Ym}
Y_M = X_M X_{M-1} \cdots X_1.
\end{equation}

When $M=1$, $Y_1 = X_1$ defines the Wishart-Laguerre unitary ensemble and it is well-known that the squared singular values of $Y_1$ form a determinantal point process with the correlation kernel expressed in terms of Laguerre polynomials. Recent studies show that the determinantal structures still hold for general $M$ \cite{Akemann-Ipsen-Kieburg13,Akemann-Kieburg-Wei13}. According to \cite{Akemann-Ipsen-Kieburg13}, the joint probability density function of the squared singular values is given by (see \cite[formula (18)]{Akemann-Ipsen-Kieburg13})
\begin{equation} \label{jpdf}
    P(x_1, \ldots, x_n) =  \frac{1}{\mathcal{Z}_n}  \Delta(x_1,\ldots,x_n)\,
        \det \left[ w_{k-1}(x_j) \right]_{j,k=1, \ldots, n},
        \qquad x_j > 0,
    \end{equation}
where the function $w_k$ is a Meijer G-function
\begin{equation} \label{wk}
    w_k(x) = \mathop{{G^{{M,0}}_{{0,M}}}\/}\nolimits\!\left({- \atop \nu_M, \nu_{M-1},  \ldots, \nu_2, \nu_1 +
    k} \Big{|} x\right),
    \end{equation}
and the normalization constant (see \cite[formula (21)]{Akemann-Ipsen-Kieburg13}) is
\[ \mathcal{Z}_n = n!\prod_{i=1}^{n}\prod_{j=0}^M \Gamma(i+\nu_j).  \]
Note that the Meijer G-function $w_k(x)$ can be written as a Mellin-Barnes integral
\begin{equation} \label{wkasMB}
    w_k(x) =  \frac{1}{2\pi i} \int_{c-i\infty}^{c+i\infty} \Gamma(s+\nu_1 + k) \prod_{j=2}^{M} \Gamma(s+\nu_j)  x^{-s} \ud s,
    \qquad k=0, 1, \ldots,
\end{equation}
with $c > 0$. As a consequence of \eqref{eq:LaginMei}, it is readily seen that if $M=1$, \eqref{jpdf} is equivalent to \eqref{eq:bioLag} with $\theta=1$.

The determinantal point process \eqref{jpdf} again is a biorthogonal ensemble. Hence, one can write the correlation kernel as
\begin{equation} \label{def:Kn}
    K_n^{\vec{\nu}}(x,y) = \sum_{k=0}^{n-1} P_k^\nu(x) Q_k^\nu(y),
    \end{equation}
where $\nu$ stands for the collection of parameters $\nu_1,\ldots,\nu_M$ and the biorthogonal functions $P_k^\nu$ and $Q_k^\nu$ are defined as follows. For each $k = 0, 1, \ldots,n-1$, $P_k^\nu$ is a monic polynomial of degree $k$ and $Q_k^\nu$ can be a linear combination of $w_0, \ldots, w_k$, uniquely defined by the orthogonality
\begin{equation} \label{PkQkbio}
    \int_0^{\infty} P_j^\nu(x) Q_k^\nu(x) \ud x = \delta_{j,k}.
    \end{equation}
In particular, we have the following explicit formulas of $P_k^\nu$ and $Q_k^\nu$ in terms of Meijer G-functions \cite{Akemann-Ipsen-Kieburg13}:
\begin{align} \label{QkMeijerG}
  Q_k^\nu(x) &=
    \frac{1}{\prod_{j=0}^M \Gamma(k+\nu_j+1)}
    \mathop{{G^{{M+1,0}}_{{1,M+1}}}\/}\nolimits\!\left({-k \atop \nu_0, \nu_{1},  \ldots, \nu_M
    } \Big{|} x \right) \nonumber  \\
    &=
    \frac{1}{2\pi i \prod_{j=0}^M \Gamma(k+\nu_j+1)}
    \int_{c-i\infty}^{c+i\infty} \frac{\prod_{j=0}^M \Gamma(s+\nu_j)}{\Gamma(s-k)} x^{-s} \ud
    s
    \end{align}
and
\begin{align} \label{PnMeijerG}
  P_n^\nu(x) & =
    -\prod_{j=0}^M\Gamma(n+\nu_j+1)\mathop{{G^{{0,1}}_{{1,M+1}}}\/}\nolimits\!\left({n+1
\atop -\nu_0, -\nu_{1},  \ldots, -\nu_{M-1}, -\nu_{M}}\Big{|}
x\right)  \nonumber
\\
&= (-1)^n \prod_{j=1}^M \frac{\Gamma(n+\nu_j+1)}{\Gamma(\nu_j+1)}
    {\; }_1 F_M \left({-n \atop 1+ \nu_1, \ldots, 1+\nu_M} \Big{|} x \right),
    \end{align}
where
 \begin{equation}\label{def:hypergeo}
 {\; }_p F_q \left({a_1,\ldots, a_p \atop b_1,\ldots,b_q} \Big{|} z \right)=\sum_{k=0}^\infty \frac{(a_1)_k\cdots (a_p)_k}{(b_1)_k \cdots (b_q)_k}\frac{z^k}{k!}
\end{equation}
is the generalized hypergeometric function with
\begin{equation}\label{eq:pochammer}
(a)_k=\frac{\Gamma(a+k)}{\Gamma(a)}=a(a+1)\cdots(a+k-1)
\end{equation}
being the Pochhammer symbol; see \eqref{eq:HyperinMei} for the second equality in \eqref{PnMeijerG}. The polynomials $P_k^\nu$ can also be interpreted as multiple orthogonal polynomials \cite{Ismail09} with respect to the first $M$ weight functions $w_j$, $j=0,\ldots,M-1$, as shown in \cite{Kuijlaars-Zhang14}. More properties of these polynomials (or in special cases) can be found in \cite{CCVA,Kuijlaars-Zhang14,Neuschel14,VA14,VAY,ZP}.

With the aid of \eqref{QkMeijerG} and \eqref{PnMeijerG}, it is shown in \cite[Proposition 5.1]{Kuijlaars-Zhang14} that the correlation kernel admits the following double contour integral representation
\begin{equation} \label{Knintegral}
    K_n^\nu(x,y) =  \frac{1}{(2\pi i)^2} \int_{-1/2-i\infty}^{-1/2+i\infty} \ud s \oint_{\Sigma}  \ud t
        \prod_{j=0}^M   \frac{\Gamma(s+\nu_j+1)}{\Gamma(t+\nu_j+ 1)}
            \frac{\Gamma(t-n+1)}{\Gamma(s-n+1)}
        \frac{x^t y^{-s-1}}{s-t},
\end{equation}
where $\Sigma$ is a closed contour going around $0, 1, \ldots, n-1$ in the positive direction and  $\Re t > -1/2$ for $t \in \Sigma$. For recent progresses in the studies of products of random matrices; see \cite{Akemann-Ipsen15}.

We point out that the kernel \eqref{Knintegral} (as well as the biorthogonal functions $P_k^\nu$ and $Q_k^\nu$) is well-defined as along as $\nu_i>-1$, $i=1,\ldots,M$, and has a random matrix interpretation if $\nu_i$ are non-negative integers, i.e., then the particles correspond to the squared singular values of the matrix $Y_M$.

\subsection{Connections between $K_n^{(\alpha,M)}$ and $K_n^\nu$}
Our final result of this paper is stated as follows.
\begin{thm}[Relating $K_n^{(\alpha,M)}$ to $K_n^\nu$]\label{thm:relation}
Let $p_k^{(\alpha,\theta)}$, $q_k^{(\alpha,\theta)}$, $P_k^\nu$ and $Q_k^\nu$ be the functions defined through biorthogonalities \eqref{eq:bioOP} and \eqref{PkQkbio}, respectively. If $\theta=M\in\mathbb{N}$, we have
\begin{equation}\label{eq:relation1}
\begin{aligned}
q_k^{(\alpha,M)}(x)&=M^{kM}P_{k}^{\tilde \nu}\left(\frac{x}{M^M} \right),\\
 x^\alpha e^{-x} p_k^{(\alpha,M)}(x)&= Q_{k}^{\tilde \nu}\left(\frac{x^M}{M^M}\right)\frac{x^{M-1}}{M^{(k+1)M-1}},
\end{aligned}
\end{equation}
where the parameter $\tilde \nu$ is given by an arithmetic sequence
\begin{equation}\label{def:tildenu}
\tilde{\nu}_j=\frac{\alpha}{M}+\frac{j}{M}-1,\quad j=1,\ldots,M.
\end{equation}
As a consequence, we have
\begin{equation}\label{eq:kerrelation}
K_n^{(\alpha,M)}(x,y)=\frac{x^{M-1}}{M^{M-1}}K_n^{\tilde \nu}\left(\frac{y^M}{M^M},\frac{x^M}{M^M}\right)
\end{equation}
where $K_n^{(\alpha,\theta)}$ and $K_n^\nu$ are two correlation kernels defined in \eqref{eq:kerBio} and \eqref{def:Kn}, respectively.
\end{thm}
\begin{proof}
Suppose now the parameters in $P_k^\nu$ are given by $\tilde \nu$ \eqref{def:tildenu}, we see from \eqref{PnMeijerG} that
\begin{align}\label{eq:Pmid}
&P_k^{\tilde \nu}(x) \nonumber \\
&=(-1)^k \prod_{j=1}^M \frac{\Gamma(k+\tilde\nu_j+1)}{\Gamma(\tilde\nu_j+1)}
    {\; }\sum_{i=0}^\infty \frac{(-k)_i}{(1+\tilde\nu_1)_i \cdots (1+\tilde\nu_M)_i}\frac{x^i}{i!} \nonumber \\
    &=(-1)^k \prod_{j=1}^M \Gamma\left(k+\frac{\alpha+j}{M}\right)
    {\; }\sum_{i=0}^k \binom{k}{i} \frac{(-x)^i}{\Gamma\left(i+\frac{\alpha+1}{M}\right) \cdots \Gamma\left(i+\frac{\alpha+M}{M}\right)},
\end{align}
where the second equality follows from the definition of Pochhammer symbol \eqref{eq:pochammer}.
In view of the Gauss's multiplication formula \cite[formula 5.5.6]{DLMF}
\begin{equation}\label{eq:Gauss}
\Gamma(nz)=(2\pi)^{(1-n)/2}n^{nz-(1/2)}\prod_{k=0}^{n-1}\Gamma\left(z+\frac{k}{n}\right)
\end{equation}
with $z=i+\frac{\alpha+1}{M}$ and $n=M$, we could further simplify \eqref{eq:Pmid} to get
\begin{align}\label{eq:Pmid2}
P_k^{\tilde \nu}=(-1)^k\sum_{i=0}^k \binom{k}{i} \frac{\Gamma(\alpha+1+kM)(-xM^M)^i}{\Gamma(\alpha+1+iM)M^{kM}}.
\end{align}
Combining \eqref{eq:Pmid2} with \eqref{eq:qexpli}, it is readily seen that
\begin{equation}\label{eq:ptoP}
 q_k^{(\alpha,M)}(x)=M^{kM}P_{k}^{\tilde \nu}\left(\frac{x}{M^M} \right),
 \end{equation}
which is the first identity in \eqref{eq:relation1}. Note that both $q_k^{(\alpha,M)}$ and $P_{k}^{\tilde \nu}$ are monic polynomials of degree $k$.

To show the second identity in \eqref{eq:relation1}, we obtain from \eqref{QkMeijerG} and \eqref{def:tildenu} that
\begin{align}\label{eq:Qk}
&Q_k^{\tilde\nu}\left(\frac{x^M}{M^M}\right) \nonumber \\
&=\frac{1}{2\pi i k! \prod_{j=1}^M \Gamma\left(k+\frac{\alpha+j}{M}\right)}
    \int_{c-i\infty}^{c+i\infty} \frac{\Gamma(s)}{\Gamma(s-k)}\prod_{j=1}^M \Gamma\left(s+\frac{\alpha}{M}-1+\frac{j}{M}\right) \left(\frac{x}{M}\right)^{-Ms} \ud s \nonumber \\
&= \frac{M^{(k+1)M}}{2\pi i k! \Gamma(\alpha+1+kM)} \int_{c-i\infty}^{c+i\infty} \frac{\Gamma(s)}{\Gamma(s-k)} \Gamma(Ms+1-M+\alpha) x^{-Ms} \ud s, \nonumber
\end{align}
where we have made use of \eqref{eq:Gauss} again in the second step. This, together with \eqref{eq:pkint2}, implies that
$$ Q_{k}^{\tilde \nu}\left(\frac{x^M}{M^M}\right)\frac{x^{M-1}}{M^{(k+1)M-1}}= x^\alpha e^{-x} p_k^{(\alpha,M)}(x),$$
as desired.

Finally, the relation \eqref{eq:kerrelation} follows immediately from a combination of \eqref{eq:kerBio}, \eqref{def:Kn} and \eqref{eq:relation1}. Alternatively, this relation can also be checked directly from the double contour integral representations \eqref{eq:Kn_double_con} and \eqref{Knintegral} with the help of multiplication formula \eqref{eq:Gauss}.

This completes the proof of Theorem \ref{thm:relation}.
\end{proof}

\begin{remark}\label{remark1}
By setting $x=y$ in \eqref{eq:kerrelation}, we simply have that $K_n^{(\alpha,M)}$ is related to $K_n^{\tilde \nu}$ via an $M$-th root transformation. Let $n\to\infty$, this in turn provides the other perspective to explain the appearance of Fuss-Catalan distribution in biorthogonal Laguerre ensembles, since it is well-known that the Fuss-Catalan distribution characterizes the limiting mean distribution for squared singular values of products of random matrices \cite{Alexeev-Gotze-Tikhomirov10,Banica-Belinschi-Capitaine-Collins11,Nica-Speicher06}. As a concrete example, we may focus on the case $\theta=M=2$. According to \cite{Kuijlaars-Zhang14,Zhang13}, the empirical measure for scaled squared singular values for the products of two Ginibre matrices converges
weakly and in moments to a probability measure over the real axis with density given by
\begin{equation}\label{eq:M2}
\frac{\sqrt{3}}{2^{\frac{4}{3}} \pi x^{\frac{2}{3}}}\left(\left(1+\sqrt{1-\frac{4x}{27}}\right)^{1/3}-\left(1-\sqrt{1-\frac{4x}{27}}\right)^{1/3}\right)
,\qquad  x\in\left(0,\frac{27}{4} \right).
\end{equation}
On the other hand, by \cite{LSZ} (see also \cite{CR14}), the limiting mean distribution for scaled particles from biorthogonal Laguerre ensembles \eqref{eq:bioLag} with $\theta=2$ takes the density function given by
\begin{equation}\label{eq:theta2}
\frac{\sqrt{3}}{2 \pi x^{\frac{1}{3}}}\left(\left(1+\sqrt{1-\frac{x^2}{27}}\right)^{1/3}-\left(1-\sqrt{1-\frac{x^2}{27}}\right)^{1/3}\right)
,\qquad  x\in\left( 0,3^{\frac{3}{2}} \right).
\end{equation}
Clearly, the density \eqref{eq:theta2} can be reduced to \eqref{eq:M2} via a change of variable $x\to 2\sqrt{x}$, as expected.                                                                                                                                           \end{remark}

\begin{remark}\label{remark2}
From \cite[Theorem 5.3]{Kuijlaars-Zhang14}, it follows that, with $K_n^\nu$ defined in \eqref{def:Kn} and $\nu_1, \ldots, \nu_M$ being fixed,
\begin{equation}
\lim_{n \to \infty}  \frac{1}{n} K_n^{\nu} \left(\frac{x}{n}, \frac{y}{n} \right) = K^{\nu}(x,y),
\end{equation}
uniformly for $x,y$ in compact subsets of the positive real axis, where
\begin{align}\label{def:K(x,y;m)} &K^{\nu} (x,y) \nonumber
\\
&=\frac{1}{(2\pi i)^2}
    \int_{-1/2-i\infty}^{-1/2+i\infty} \ud s \int_{\Sigma}  \ud t     \prod_{j=0}^M   \frac{\Gamma(s+\nu_j+1)}{\Gamma(t+\nu_j+ 1)}
        \frac{\sin \pi s}{\sin \pi t} \frac{x^t y^{-s-1}}{s-t}
   \nonumber  \\
&=\int_0^1  G^{1,0}_{0,M+1} \left(
\begin{array}{c} - \\-\nu_0, -\nu_1, \ldots, -\nu_M
\end{array}\Big{|} ux \right)
    G^{M,0}_{0,M+1} \left( \begin{array}{c} - \\  \nu_1, \ldots, \nu_M,\nu_0
\end{array}\Big{|}uy
    \right) \ud u,
\end{align}
and where $\Sigma$ is a contour starting from $+\infty$ in the upper half plane and returning to $+\infty$ in the lower half plane which
encircles the positive real axis and $\Re t>-1/2$ for $t\in\Sigma$. This fact, together with our relation \eqref{eq:kerrelation} and the hard edge scaling limits of Borodin \eqref{eq:BorHard}, implies that
\begin{equation}\label{eq:BorKZ}
M x^{\alpha}\int_0^1J_{\frac{\alpha+1}{M},\frac{1}{M}}(ux)J_{\alpha+1,M}(\left(uy\right)^{M})u^\alpha\ud u
 = \frac{x^{M-1}}{M^{M-1}}K^{\tilde \nu}\left(\frac{y^M}{M^M},\frac{x^M}{M^M}\right).
\end{equation}
The identity \eqref{eq:BorKZ} was first proved in \cite{Kuijlaars-Stivigny14}, where the authors gave a direct proof by noting that Wright generalized Bessel functions $J_{a,b}$ defined in \eqref{eq:Wright} can be
expressed in Meijer G-functions if $b$ is a rational number. Since it is easily seen from \eqref{def:K(x,y;m)} and the multiplication formula \eqref{eq:Gauss} that
\begin{equation}
\frac{x^{M-1}}{M^{M-1}}K^{\tilde \nu}\left(\frac{y^M}{M^M},\frac{x^M}{M^M}\right)=K^{(\alpha,M)}(x,y), \end{equation}
where $K^{(\alpha,M)}(x,y)$ is given in \eqref{eq:hardedgelim}, the proof presented in \cite{Kuijlaars-Stivigny14} also gives a direct proof of identity \eqref{eq:BorZhang} if $\theta=M\in\mathbb{N}$. To show \eqref{eq:BorZhang} for general $\theta\geq 1$, we first observe from \eqref{def:q}, the residue theorem and \eqref{eq:Wright} that
\begin{equation}
q^{(\alpha,\theta)}(x)=\sum_{k=0}^{\infty}\frac{(-1)^k}{k!}\frac{\theta x^{k\theta}}{\Gamma(\alpha+1+k\theta)}
=\theta J_{\alpha+1,\theta}(x^\theta).
\end{equation}
Similarly, by deforming the vertical line in \eqref{def:p} to be a loop starting from $-\infty$ in the lower
half plane and returning to $-\infty$ in the upper half plane which encircles the negative real axis, we again obtain from the residue theorem that
\begin{equation}
p^{(\alpha,\theta)}(x)=\sum_{k=0}^{\infty}\frac{(-1)^k}{k!}\frac{x^{\alpha+k}}{\Gamma\left(\frac{\alpha+1+k}{\theta}\right)}
=x^\alpha J_{\frac{\alpha+1}{\theta},\frac{1}{\theta}}(x).
\end{equation}
A combination of the above two formulas, \eqref{eq:hardedgelim} and \eqref{eq:hardedgelim2} gives us \eqref{eq:BorZhang}.
\end{remark}

\appendix

\section{The Meijer G-function}
For convenience of the readers, we give a brief introduction to the Meijer G-function in this
appendix, which includes its definition and some properties used in this paper.

By definition, the Meijer G-function is given by the
following contour integral in the complex plane:
\begin{equation}\label{def:Meijer}
G^{m,n}_{p,q}\left({a_1,\ldots,a_p \atop b_1,\ldots,b_q}\Big{|}
z\right)
=\frac{1}{2\pi i}\int_\gamma
\frac{\prod_{j=1}^m\Gamma(b_j+u)\prod_{j=1}^n\Gamma(1-a_j-u)}
{\prod_{j=m+1}^q\Gamma(1-b_j-u)\prod_{j=n+1}^p\Gamma(a_j+u)}z^{-u}
\ud u,
\end{equation}
where $\Gamma$ denotes the usual gamma function and the branch cut
of $z^{-u}$ is taken along the negative real axis. It is also
assumed that
\begin{itemize}
  \item $0\leq m\leq q$ and $0\leq n \leq p$, where $m,n,p$ and $q$
  are integer numbers;
  \item The real or complex parameters $a_1,\ldots,a_p$ and
  $b_1,\ldots,b_q$ satisfy the conditions
  \begin{equation*}
  a_k-b_j \neq 1,2,3, \ldots, \quad \textrm{for $k=1,2,\ldots,n$ and $j=1,2,\ldots,m$,}
  \end{equation*}
  i.e., none of the poles of $\Gamma(b_j+u)$, $j=1,2,\ldots,m$ coincides
  with any poles of $\Gamma(1-a_k-u)$, $k=1,2,\ldots,n$.
\end{itemize}
The contour $\gamma$ is chosen in such a way that all the poles of
$\Gamma(b_j+u)$, $j=1,\ldots,m$ are on the left of the path, while
all the poles of $\Gamma(1-a_k-u)$, $k=1,\ldots,n$ are on the right,
which is usually taken to go from $-i\infty$ to $i\infty$. For more details, we refer to
the references \cite{Luke,DLMF}.

Most of the known special functions can be viewed as special cases of the
Meijer G-functions. For instance, with the generalized hypergeometric function ${\; }_p F_q$ given in \eqref{def:hypergeo}, one has \cite[formula 16.18.1]{DLMF}
\begin{equation}\label{eq:HyperinMei}
\mathop{{{}_{p}F_{q}}\/}\nolimits\!\left({a_{1},\dots,a_{p}\atop b_{1},\dots,b%
_{q}}\Big{|}z\right)=\frac{\prod\limits_{k=1}^{q}\Gamma(b_k)}{\prod\limits_{k=1}^{p}\Gamma(a_k)}
G^{1,p}_{p,q+1}\left({1-a_1,\ldots,1-a_p \atop 0, 1-b_1,\ldots,1-b_q}\Big{|}
-z\right).
\end{equation}
This, together with the fact that
   \begin{equation}\label{eq:multiply}
   z^{\alpha}G^{m,n}_{p,q}\left({a_1,\ldots,a_p \atop b_1,\ldots,b_q}\Big{|}z\right)=G^{m,n}_{p,q}\left({a_1+\alpha,\ldots,a_p+\alpha \atop b_1+\alpha,\ldots,b_q+\alpha}\Big{|}z\right),
   \end{equation}
gives us
\begin{equation}\label{eq:LaginMei}
  x^\alpha e^{-x}=G^{1,0}_{0,1}\left({- \atop \alpha}\Big{|}x\right)=\frac{1}{2\pi i}\int_\gamma \Gamma(\alpha+s)x^{-s}\ud s.
  \end{equation}

\section*{Acknowledgment}
The author thanks Peter Forrester and Dong Wang for helpful communications and for providing me with an early copy of the preprint~\cite{Forrester-Wang15} on a related study of the Laguerre biorthogonal ensemble upon completion
of the present work. The author also thanks the anonymous referees for their careful reading and constructive suggestions.

This work is partially supported by The Program for Professor of Special Appointment (Eastern Scholar) at Shanghai Institutions of Higher Learning (No. SHH1411007) and by Grant EZH1411513 from Fudan University.




\end{document}